\documentclass[prx,reprint,floatfix,letterpaper,twocolumn,nofootinbib,showpacs,longbibliography,superscriptaddress]{revtex4-1}
\usepackage[caption=false]{subfig}
\usepackage{graphicx} 
\usepackage{epstopdf}
\usepackage{array}
\usepackage{verbatim}
\usepackage{amsmath,amsfonts,amssymb,amscd}
\usepackage{amsthm}
\usepackage{tabularx}
\usepackage{stmaryrd}
\usepackage{enumerate}
\usepackage[ruled]{algorithm2e}
\usepackage{enumitem}
\usepackage{wasysym}
\usepackage[dvipsnames]{xcolor}
\usepackage[normalem]{ulem}

\input{Qcircuit}

\usepackage{hyperref}
\hypersetup{
    bookmarksnumbered=true, 
    unicode=false, 
    pdfstartview={FitH}, 
    pdftitle={}, 
    pdfauthor={}, 
    pdfsubject={}, 
    pdfcreator={}, 
    pdfproducer={}, 
    pdfkeywords={}, 
    pdfnewwindow=true, 
    colorlinks=true, 
    linkcolor=blue, 
    citecolor=blue, 
    filecolor=blue, 
    urlcolor=blue 
}

\newcounter{ex}

\theoremstyle{plain}
\newtheorem{thm}{Theorem}
\newtheorem{cor}[thm]{Corollary}
\newtheorem{lem}[thm]{Lemma}

\newtheorem{example}[ex]{Example}

\theoremstyle{definition}
\newtheorem{defn}[thm]{Definition}

\newcommand{\eq}[1]{(\hyperref[eq:#1]{\ref*{eq:#1}})}

\renewcommand{\sec}[1]{\hyperref[sec:#1]{Section~\ref*{sec:#1}}}
\newcommand{\thrm}[1]{\hyperref[thm:#1]{Theorem~\ref*{thm:#1}}}
\newcommand{\lemm}[1]{\hyperref[lemm:#1]{Lemma~\ref*{lemm:#1}}}
\newcommand{\prop}[1]{\hyperref[prop:#1]{Proposition~\ref*{prop:#1}}}
\newcommand{\corr}[1]{\hyperref[corr:#1]{Corollary~\ref*{corr:#1}}}
\newcommand{\fig}[1]{\hyperref[fig:#1]{Figure~\ref*{fig:#1}}}

\DeclareMathAlphabet{\matheu}{U}{eus}{m}{n}

\newcommand{\supp}[1]{\mathrm{supp}(#1)}

\newcolumntype{L}[1]{>{\raggedright}p{#1}}
\newcolumntype{C}[1]{>{\centering}p{#1}}
\newcolumntype{R}[1]{>{\raggedleft}p{#1}}
\newcolumntype{D}{>{\centering\arraybackslash}X}

\definecolor{darkgreen}{rgb}{0,0.5,0}
\definecolor{darkblue}{rgb}{0,0,0.5}

\newcommand{\dis}[2]{\Delta_{#1}(#2)}
\newcommand{\disj}{\Delta}

\newcommand{\dmin}{d_\downarrow}
\newcommand{\dmax}{d_\uparrow}

\newcounter{thm_counter}

\begin{document}
\title{The disjointness of stabilizer codes and limitations on fault-tolerant logical gates
}
\author{Tomas Jochym-O'Connor}
\affiliation{\small Walter Burke Institute for Theoretical Physics}
\affiliation{\small Institute for Quantum Information \& Matter, California Institute of Technology, Pasadena, CA 91125, USA}
\author{Aleksander Kubica}
\affiliation{\small Institute for Quantum Information \& Matter, California Institute of Technology, Pasadena, CA 91125, USA}
\affiliation{\small Perimeter Institute for Theoretical Physics, Waterloo, ON N2L 2Y5, Canada}
\affiliation{\small Institute for Quantum Computing, University of Waterloo, Waterloo, ON N2L 3G1, Canada}
\author{Theodore J. Yoder}
\affiliation{\small Department of Physics, Massachusetts Institute of Technology, Cambridge, MA 02139, USA}

\begin{abstract}
Stabilizer codes are a simple and successful class of quantum error-correcting codes. Yet this success comes in spite of some harsh limitations on the ability of these codes to fault-tolerantly compute. Here we introduce a new metric for these codes, the disjointness, which, roughly speaking, is the number of mostly non-overlapping representatives of any given non-trivial logical Pauli operator. We use the disjointness to prove that transversal gates on error-detecting stabilizer codes are necessarily in a finite level of the Clifford hierarchy. We also apply our techniques to topological code families to find similar bounds on the level of the hierarchy attainable by constant depth circuits, regardless of their geometric locality. For instance, we can show that symmetric 2D surface codes cannot have non-local constant depth circuits for non-Clifford gates.
\end{abstract}

\maketitle

\section{Introduction}

Quantum error-correcting codes form the foundation of scalable quantum computing \cite{Shor1995,Steane1996,Preskill1998}.
By construction, quantum codes serve as quantum memories by protecting encoded data from a noisy environment and successfully extending the storage time, at least if the noise is sufficiently small.
However, a quantum computer should do more than just store quantum data; it needs to also apply logical operations to the data \cite{Gottesman1998,Kitaev2003}.
These operations must therefore be implemented fault-tolerantly upon quantum codes.

Generally, operators are fault-tolerant if they do not couple too many qubits within a particular codeblock. This condition is sufficient to limit the spread of errors and also guarantee that if parts of the circuitry implementing the operator were to fail that not many qubits would be affected. With respect to some partitioning of the code qubits into small, disjoint subsets $Q_i$, a transversal operator acts on each subset of qubits $Q_i$ independently. For a family of codes with increasing size, a constant-depth logical operator is implementable by a constant (independent of the code size) depth circuit over the subsets $Q_i$. Transversal and constant depth circuits are some of the simplest possible fault-tolerant operators both theoretically and experimentally, so it is important to understand exactly what logical operators they can implement.

Unfortunately, the set of transversal or, more generally, constant-depth logical operators is inherently limited, with computational universality generally incommensurate with the error-correction capabilities of the code.
In particular, there is a no-go theorem due to Eastin and Knill which states that transversal operators on any non-trivial quantum code belong to a finite group, and thus cannot be universal \cite{EK09, ZCC11}.
Similar no-go theorems limiting logical operators to be in a finite level of the Clifford hierarchy were derived for transversal single-qubit gates and two-qubit diagonal gates on stabilizer codes \cite{AJ16}, as well as for constant-depth, local circuits on stabilizer and subsystem topological codes \cite{Bravyi2013,Pastawski2015}.
The latter result has an important implication --- one cannot achieve a universal gate set with constant-depth local circuits on two-dimensional (2D) topological codes such as those in \cite{Dennis2002, Bombin2006}.
We also remark that one can consider more general models beyond stabilizer codes, such as 2D topological quantum field theories, and characterize the set of gates implementable by locality-preserving unitaries \cite{beverland2016protected,Webster2017}.

Here we address several related questions regarding transversal and constant depth logical operators on stabilizer codes using a new quantity called the disjointness of the code.
The disjointness, roughly speaking, is the number of mostly non-overlapping representatives of any given non-trivial logical Pauli operator.
We use the disjointness to show that all transversal logical operators on stabilizer codes must be in the Clifford hierarchy, as conjectured by Zeng et al.~\cite{ZCC11}. Moreover, we find explicit upper bounds on the level attainable.
Importantly, our result, when applied to families of codes of growing size, restricts constant depth circuits to the Clifford hierarchy, regardless of their geometric locality.
For instance, for the 2D toric code on a square lattice of size $O(l)\times O(l)$ we find that even non-local constant depth circuits cannot implement logical non-Clifford operators.
Asymmetry of logical operators appears in our bounds as a necessary condition for possessing constant depth circuits for non-Clifford gates, such as those on 3D color and toric codes \cite{Bombin15,Kubica2015b} and on asymmetric 2D Bacon-Shor codes \cite{Yoder2017}.

\section{The intuition}\label{sec:intuition}
In this section, we sketch out the proof that constant depth circuits, even with gates that are geometrically non-local, cannot implement logical non-Clifford operators on the 2D toric code of size $O(l)\times O(l)$, see Fig.~\ref{fig:2DTC}.
We use the following two key ideas: (i) there are many non-overlapping representatives for logical Pauli $\bar X$ and $\bar Z$ operators, (ii) logical operators supported on a correctable region are trivial.

In order to find out what logical gate a unitary $U$ implements, it is sufficient to characterize the action of $U$ on the logical Pauli operators.
Let $[A,B]=ABA^\dag B^\dag$ represent the group commutator of two unitaries $A$ and $B$ \cite{Bravyi2013}.
We know that for any two logical Pauli operators $\bar P,\bar Q\in \{\bar X, \bar Z\}$, if the group commutator $[[U,\bar P],\bar Q]$ is a trivial logical operator, then the unitary $U$ implements a logical Clifford operator.\footnote{Since we restrict unitary operators to the Clifford group, it is sufficient to consider generators $\bar X$ and $\bar Z$ of the logical Pauli group. However, restricting operators to levels of the Clifford hierarchy beyond the third requires considering all logical Paulis.}

Let us pick a representative $p$ of the logical operator $\bar P \in \{ \bar X, \bar Z \}$, such that $|\supp p| = O(l)$.
We denote by $\supp A$ the set of qubits an operator $A$ acts on non-trivially (we will later generalize this notion).
Since we assume that $U$ is constant depth, then $|\supp{[U,p]}| = O(l)$.
Note that a tensor product of Pauli $Z$ operators on qubits along any vertical path on the lattice would implement the logical Pauli $\bar Z$, see Fig.~\ref{fig:2DTC}.
Similarly, Pauli $X$ operators along any horizontal path implement the logical Pauli $\bar X$.
Thus, for any operator $\bar Q \in \{ \bar X, \bar Z \}$, we can choose $O(l)$ different, non-overlapping representatives.
Using the pigeonhole principle, we are guaranteed to find a representative $q$ of $\bar{Q}$, such that it has constant overlap with $[U,p]$.
This, in turn, implies that the operator $[[U,p],q]$ is supported on a constant-size region, $|\supp{[[U,p],q]}| = O(1)$.
Since the distance of the code is $O(l)$, the region $\supp{[[U,p],q]}$ is correctable.
We conclude that $[[U,p],q]$ can only be a trivial logical operator, and thus $U$ implements a logical Clifford operator.

\begin{figure}
\centering
\subfloat[]{\includegraphics[width= 0.22\textwidth,trim={11cm 8cm 11cm 4cm},clip]{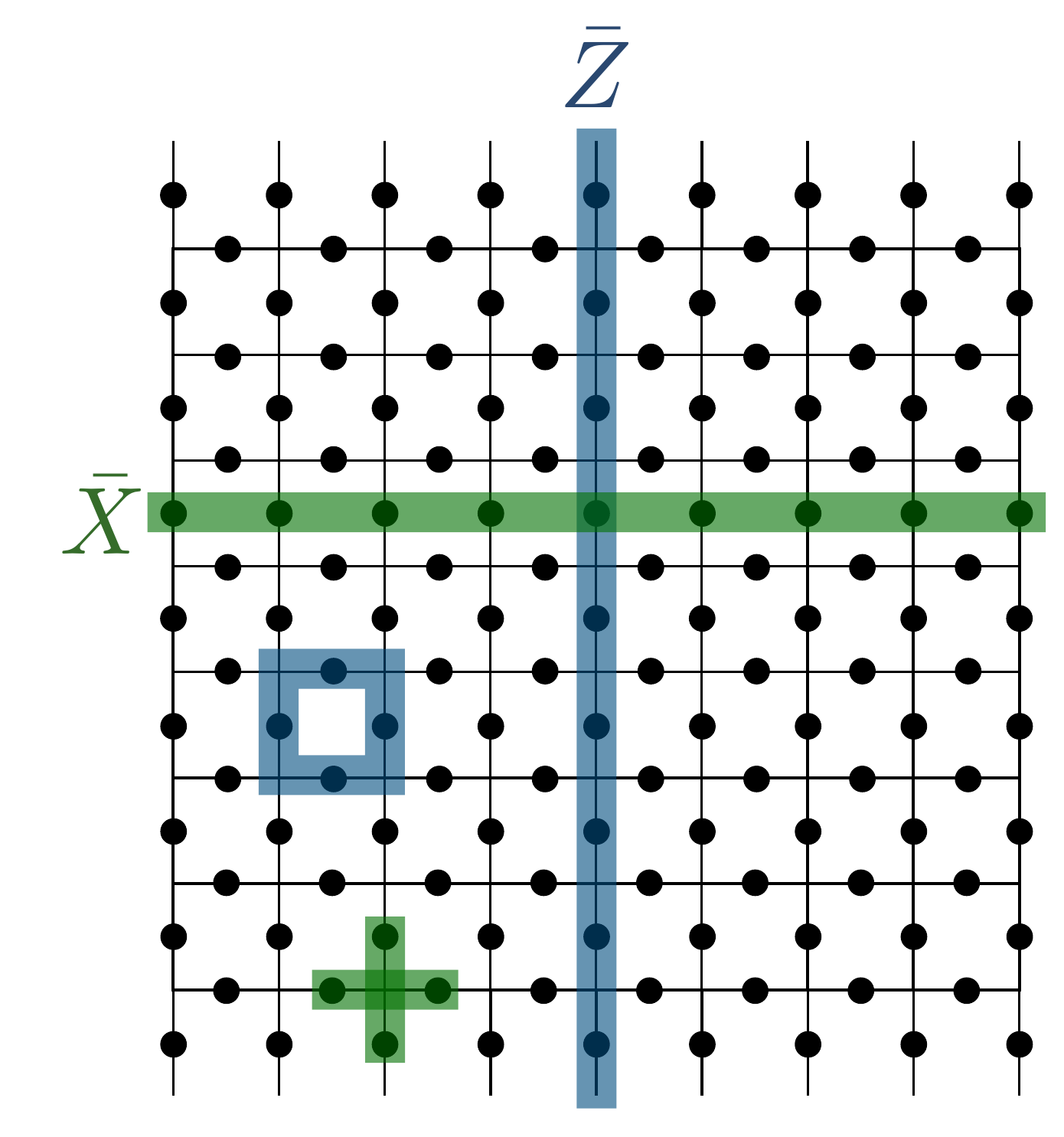}}
\hspace{1em}
\subfloat[]{\includegraphics[width= 0.22\textwidth,trim={11cm 8cm 11cm 4cm},clip]{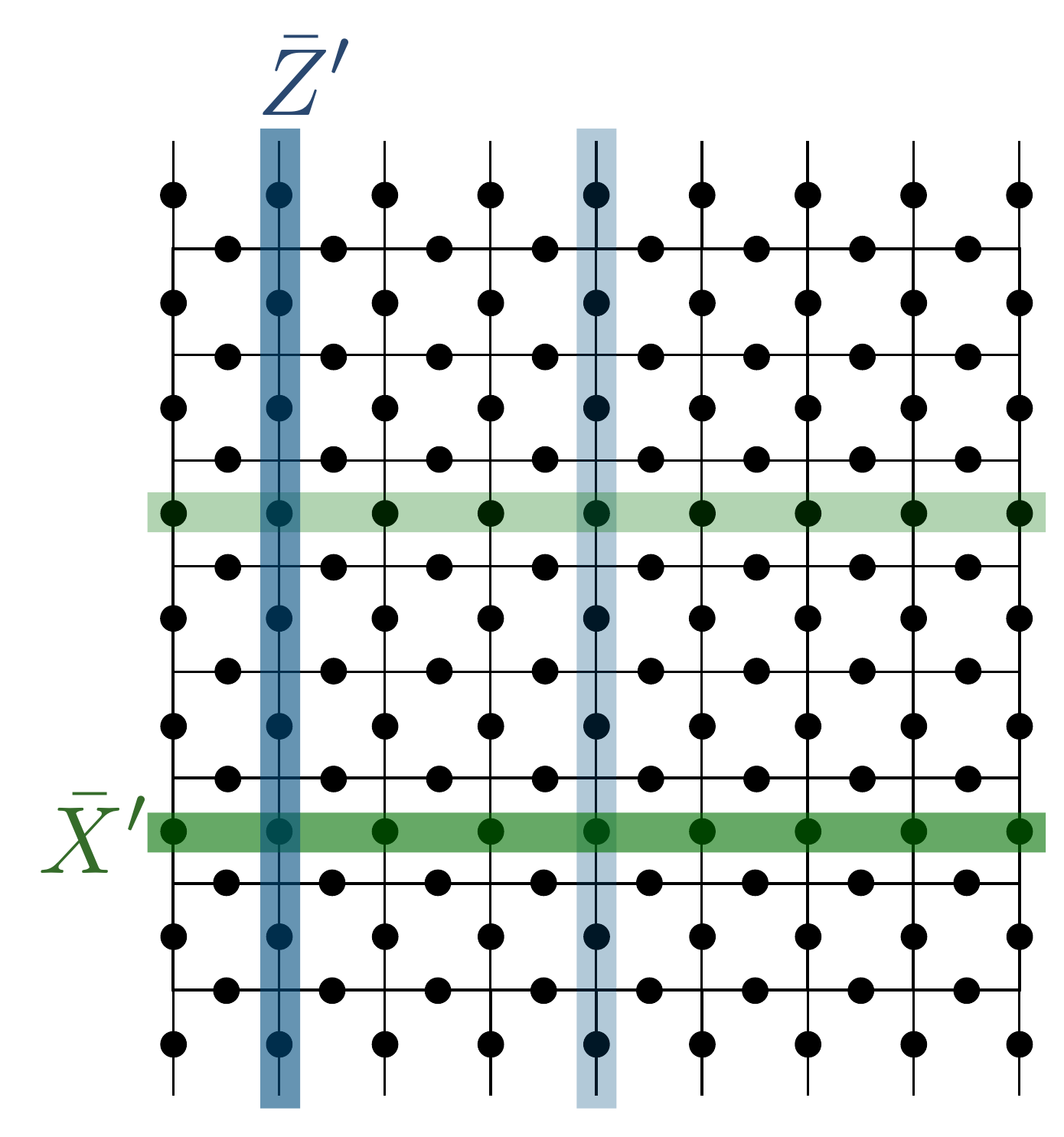}}\\
\caption{Logical Pauli string operators of the 2D surface code. In~(a), $X$-vertex and $Z$-plaquette stabilizers are shown along with a choice of Pauli logical operators. In~(b), a different, equivalent, choice for the logical Pauli operators is shown.}
\label{fig:2DTC}
\end{figure}

The property of any stabilizer code which we would like to abstract from the provided example of the 2D toric code is the existence of several (mostly) disjoint representatives of the same logical Pauli operator.
In the following sections, we will introduce a notion of disjointness of a stabilizer code, which quantitatively captures that property.
We remark that the disjointness of the 2D toric code is $O(l)$, since we can find a set of $O(l)$ non-overlapping representatives of $\bar X$ or $\bar Z$.

\section{Preliminaries}

Let us consider systems composed of $m$-dimensional qudits, $m\geq 2$.
The Pauli group on a set of $n$ qudits, denoted $\mathcal{P}_n$, is generated by the $X$- and $Z$-type operators\footnote{Qubit Paulis $m=2$ are traditionally defined to be generated by $X,Z$, and also $Y=iXZ$. Doing so does not change our results.}
\begin{equation}
X =\sum_{l=0}^{m-1}\ket{l\oplus 1}\bra{l},\quad Z =\sum_{l=0}^{m-1}\omega^l_m\ket{l}\bra{l},
\end{equation}
where addition $\oplus$ inside bra-kets is modulo $m$ and $\omega_m=\exp(2\pi i/m)$. Letting $\mathcal{U}_n$ denote the group of $n$-qudit unitaries, we note that  $\mathcal{P}_n$ is a subgroup of $\mathcal{U}_n$ because $X,Z\in\mathcal{U}_n$.

Any Pauli group $\mathcal{P}$ can be used to define a hierarchy of $n$-qudit unitaries called the Clifford hierarchy \cite{Gottesman1999b}.
The $M^{\text{th}}$ level of this hierarchy is a finite set of unitaries (if the global phases are ignored) recursively defined as
\begin{align}
C_1(\mathcal{P})&=\mathcal{P},\\\label{eq:hiearchy_alt}
C_M(\mathcal{P})&=\{U\in\mathcal{U}_n: [U,p] \in C_{M-1}(\mathcal{P}), \forall p\in\mathcal{P}\}.
\end{align}
The first and second levels of the hierarchy correspond to the Pauli and Clifford groups, respectively.

In this article, we focus our attention on a particularly popular class of quantum codes --- stabilizer codes \cite{Gottesman1997}.
A stabilizer code is defined by the stabilizer group
$\mathcal{S}=\langle s_1,s_2,\dots,s_{n-k}\rangle\subseteq\mathcal{P}_n$,
which is generated by $n-k$ mutually commuting Pauli operators.
The codespace $\mathcal{C}$ is a subspace of the Hilbert space $\mathcal{H}\simeq (\mathbb{C}^m)^{\otimes n}$ on $n$ qudits, which is the simultaneous $(+1)$-eigenspace of all stabilizer generators $s_i$.
We denote by $\llbracket n,k\rrbracket$ a qudit stabilizer code, which uses $n$ physical qudits to encode $k$ logical ones.

For any stabilizer code, a logical operator is a unitary on the Hilbert space $\mathcal{H}$ that maps states in $\mathcal{C}$ to states in $\mathcal{C}$.
In particular, logical Pauli operators can be found as elements of the normalizer $\mathcal{N}(\mathcal{S})$ of the stabilizer group $\mathcal{S}$ in the Pauli group $\mathcal{P}_n$.
We choose $2k$ generators $\bar X_i,\bar Z_i \in \mathcal{P}_n$ of the logical Pauli group $\overline{\mathcal{P}_k}$ that commute with all stabilizer generators, as well as satisfy
\begin{align}
[\bar X_i,\bar Z_j]&=\omega_m^{-\delta_{ij}}I,\\
[\bar X_i,\bar X_j]&=[\bar Z_i,\bar Z_j]=I.
\end{align}
We define $\mathcal{L}$ to be the set of sets of all non-trivial logical Pauli operators as follows
\begin{equation}
\mathcal{L} = \left\{\mathcal{S}\prod_{i=1}^k\bar X_i^{a_i} \bar Z_i^{a_{i+k}}:a\in\{0,1,\ldots, m-1 \}^{2k} \setminus \{ 0\}^{2k}\right\}.
\end{equation}
We remark that each element $G\in\mathcal{L}$ is a coset of $\mathcal{S}$ in $\mathcal{N}(\mathcal{S})$, although in examples we abuse notation and equate $G$ with the logical Pauli it corresponds to (e.g.~$\bar X=\mathcal{S}\bar X\in\mathcal{L}$).
Also, $G$ contains $|\mathcal{S}|=m^{n-k}$ representatives of the same non-trivial logical operator.

\section{Transversal gates}

All the logical operators we implement should be fault-tolerant, in a sense that they do not spread errors throughout the system in an uncontrollable way.
The simplest example of such an operator is a transversal logical operator $U$. Typically, when one says $U$ is transversal, it means that $U$ is a tensor product of single-qudit unitaries. However, we consider a more general definition of a transversal gate\footnote{Our definition is nevertheless still consistent with the definition of Eastin and Knill \cite{EK09}.}. Partition the set of $n$ physical qudits, labeled by integers from $[n]=\{1,2,\dots,n\}$, into $N$ disjoint, non-empty subsets $Q_i\subseteq [n]$, namely
\begin{equation}
[n] = Q_1 \cup Q_2 \cup \ldots \cup Q_N.
\end{equation}
Then, we say that an $n$-qudit unitary $U$ is transversal if it can be decomposed as $U = \bigotimes_{i=1}^N U_i$, where each unitary $U_i$ acts only on qudits in the subset $Q_i$.
The support of $U$, denoted by $\supp U\subseteq [N]$, is the index set of all subsets $Q_i$, on which $U$ acts non-trivially. The typical notion of transversal gate now simply corresponds to the partition into single-qudits, $Q_i=\{i\}$.

We emphasize that for a given code, the set of transversal logical operators can depend on the choice of the qudit partition. In particular, if the partition is not fixed, then one can achieve a universal gate set of transversal operators, as in the following example \cite{JL14}.
\begin{example}\label{ex:105_code}
Consider the $[[105,1]]$ code, which is a concatenation of the Steane $7$-qubit code with the $15$-qubit Reed-Muller code. We illustrate this code in Fig.~\ref{fig_partition} as a $7\times 15$ array of qubits.
We consider two qubit partitions: (a) each $Q_i$ is a subset of $7$ qubits from the $i^{\textrm{th}}$ column, (b) each $Q_i$ is a subset of $15$ qubits from the $i^{\textrm{th}}$ row.
With respect to the first and second partitions, the $[[105,1]]$ code has, correspondingly, transversal logical $T = \mathrm{diag}(1,e^{2\pi i/8})$ and Hadamard gates.
For more details, see~\cite{JL14}.
\end{example}
\noindent In contrast, we fix a partition and prove limitations on logical operators with respect to that partition. For instance, in this fixed-partition scenario, \cite{EK09} implies that the group of transversal operators is finite and therefore not universal.

\begin{figure}
\centering
\includegraphics[width= \columnwidth,]{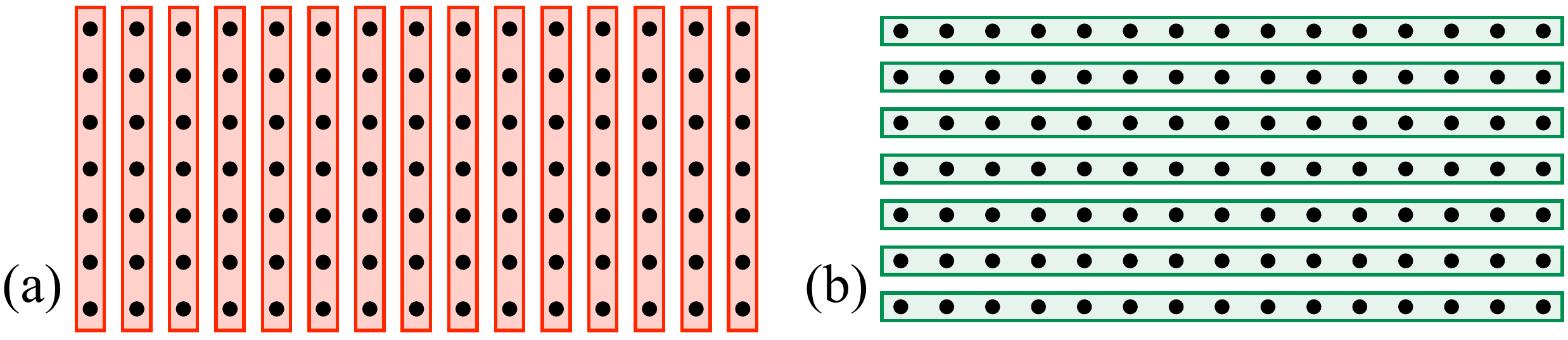}
\caption{Two different partitions of the $[[105,1]]$ qubit stabilizer code. Depending on the partition, the code can either have a transversal logical (a) $T = \mathrm{diag}(1,e^{2\pi i/8})$ gate or (b) Hadamard gate.}
\label{fig_partition}
\end{figure}

Transversal unitaries are a special case of what we call $q$-local operators of depth $h$ (with respect to the partition $\{Q_i\}$). A unitary $U$ is $q$-local of depth one, if it is transversal with respect to a second, ``coarse-grained'' partition $\{R_j\}$, where each $R_j$ is the union of at most $q$ of the $Q_i$. Accordingly, a $q$-local unitary of depth $h$ is a product of $h$ $q$-local unitaries of depth one. We note that transversal operators are 1-local unitaries of depth one.

\section{Distance and disjointness}

A fundamental property of stabilizer codes is the distance.
Typically, one says that the code has distance $d$ if it can detect any error which affects at most $d-1$ qudits.
Here, however, we consider distance with respect to the qudit partition $\{Q_i\}$.
First, we define the distance $d(G)$ of the non-trivial logical operator $G\in\mathcal{L}$ to be the size of the smallest support of any of its representatives:
\begin{equation}
d(G) = \min_{g\in G} |\supp g|.
\end{equation}
Then, we introduce two notions of the distance $\dmin$ and $\dmax$, the min- and max-distance of the code, as follows,
\begin{eqnarray}
\dmin = \min_{G\in \mathcal{L}} d(G),\quad\quad \dmax = \max_{G\in \mathcal{L}} d(G).
\end{eqnarray}
We call a code error-detecting iff its min-distance $\dmin$ is greater than one.
Note that the min-distance $\dmin$ is never greater than the (standard) distance $d$ of the code.
Also, if we choose a single-qudit partition, then those two quantities coincide, $\dmin = d$.

In this article, we propose a new quantity for quantum stabilizer codes, the disjointness, which proves remarkably useful for establishing limitations on logical gates. First, for any non-trivial logical operator $G\in\mathcal{L}$ and a positive integer $c\geq 1$ we define $c$-disjointness $\dis c G$ to be the maximal number (divided by $c$) of representatives of $G$ chosen in such a way that at most $c$ representatives have support on any $Q_i$, a subset of the qudit partition:
\begin{eqnarray}
\nonumber
\dis c G = \frac{1}{c} \max_{A\subset G} \{ |A|:&&\textrm{ at most $c$ elements $a\in A$}\quad\\
&&\textrm{ have support on any $Q_i$} \}.\quad
\label{eq_disjointness_c}
\end{eqnarray}
We call the set $A$ in Eq.~(\ref{eq_disjointness_c}) $c$-disjoint. To build intuition about the $c$-disjointness consider a small example.
\begin{example}
Consider the $\llbracket4,2\rrbracket$ qubit code with the stabilizer group $\mathcal{S}=\langle X^{\otimes4},Z^{\otimes4}\rangle$ and the single-qubit partition. There are four equivalent logical operators implementing a logical $\bar X_1=X_1X_2$, which form a set
\begin{equation}
G=\{X_1X_2, X_3X_4, Y_1Y_2Z_3Z_4, Z_1Z_2Y_3Y_4\}.
\end{equation}
The set $\{X_1X_2, X_3X_4\}$ is a maximal $1$-disjoint set, $\{X_1X_2, X_3X_4, Y_1Y_2Z_3Z_4\}$ is a maximal $2$-disjoint set, and $G$ itself is a maximal $3$-disjoint set. Thus, $\dis 1 G = 2$, $\dis 2 G =3/2$, $\dis 3 G =4/3$.
\end{example}

Now, we are ready to define the disjointness $\disj$ of a code.
\begin{defn}[disjointness]
\label{defn:disjointness}
For any $n$-qudit stabilizer code with the set of non-trivial logical operators $\mathcal{L}$ and a qudit partition $[n] = Q_1\cup Q_2 \cup \ldots \cup Q_N$,  the disjointness is defined as
\begin{equation}\label{eq:defn_disjointness}
\disj = \max_{c \geq 1} \min_{G\in \mathcal{L}} \dis c G 
\end{equation}
\end{defn}
\noindent We illustrate disjointness with the following example of the 2D surface code.
\begin{example}\label{ex:surface_code}
Consider the 2D surface code of size $l\times l$ encoding one logical qubit \cite{Bravyi1998} and the single-qubit partition. We have $\dmax=d(\bar Y)=2l-1$ and $\dmin=d(\bar X)=d(\bar Z)=l$. Moreover, there are exactly $l$ representatives of $\bar X$ with weight $l$, and they are all disjoint. Thus, $\dis 1 {\bar X}=l$. Similarly, $\dis 1 {\bar Z}=l$. In contrast, different representatives of $\bar Y$ necessarily overlap, but we can nevertheless find $l$ representatives of minimal weight $2l-1$, such that each qubit is in the support of at most two of them. Those representatives of $\bar Y$ form a 2-$disjoint$ set. Thus, $\dis 1 {\bar Y}=1$, but $\dis 2 {\bar Y}=l/2$. We conclude that the disjointness $\disj$ of the surface code satisfies $\disj\ge l/2$.
\end{example}

The disjointness $\disj$ turns out to be an important quantity characterizing stabilizer codes. In particular, we use it to find bounds on the level of the logical Clifford hierarchy achievable with transversal (see Theorem~\ref{thm:finite_code_bound} in Section~\ref{sec:limit_trans}) or constant-depth (see Theorem~\ref{thm:deep_circ_code_bound} in Section~\ref{sec:bounding_constant_depth_gates}) logical unitaries.
To facilitate further discussion, we present key properties of the disjointness.
\begin{lem}[properties of disjointness]
\label{lem:properties_of_disjointness}
For any $\llbracket n,k\rrbracket$ stabilizer code and any partition $[n]=Q_1\cup Q_2\cup\dots\cup Q_N$, the disjointness satisfies
\begin{enumerate}[label=(\roman*)]
\item $1\le \disj \le \min(\dmin,N/\dmax)$
\item $\disj > 1$ iff the stabilizer code is error-detecting, i.e., $\dmin >1$.
\end{enumerate}
\end{lem}
\begin{proof}
We begin by proving four bounds on $c$-disjointness that together imply (i). In particular, let $G,G'\in\mathcal{L}$ be two non-commuting, non-trivial logical operators. That is, $[g,g']\neq I$ for all $g\in G$ and $g'\in G'$. Then, for any $1\leq c \leq m^{n-k}$ (recall $m$ is the qudit dimension),
\begin{align}\label{eq:triv_bound}
1\le\dis c G&\le m^{n-k}/c,\\\label{eq:dmin_bound}
\dis c G &\le d(G'),\\\label{eq:dmax_bound}
\dis c G d(G)&\le N.
\end{align}
Moreover, each upper bound holds for all $c\ge1$.

The lower bound in Eq.~\eqref{eq:triv_bound} is true because any $c\le m^{n-k}$ elements of $G$ form a $c$-disjoint set of size $c$. The upper bound in Eq.~\eqref{eq:triv_bound} results because any $c$-disjoint set $A\subseteq G$ satisfies $|A|\le |G|=m^{n-k}$. As a result of the upper bound, for any $c>m^{n-k}$, $\dis c G<1$. Along with the lower bound, this implies $\min_{c\ge1}\dis c G=\min_{1\le c\le m^{n-k}}\dis c G$ for all $G\in\mathcal{L}$, which simplifies the definition of disjointenss Eq.~\eqref{eq:defn_disjointness}.

For Eq.~\eqref{eq:dmin_bound}, choose a maximal $c$-disjoint set $A\subseteq G$ and a representative $g'\in G'$ of minimal support. That is, $|A|=c\dis c G$ and $|\supp{g'}|=d(G')$. By definition, every $g\in A$ does not commute with $g'$. Thus, $g$ and $g'$ have to have non-trivial overlap, $|\supp{g}\cap\supp{g'}|\ge 1$. 
Consider any collection $H\subseteq [N]$ of some qudit subsets $Q_i$. 
Since at most $c$ elements of $A$ intersect at any subset of qudits $Q_i$, we have the inequality
\begin{eqnarray}\label{eq:total_c_disjoint_intersection_0}
\sum_{g\in A}|\supp{g}\cap H| &=& \sum_{g\in A}\sum_{i\in H}|\supp{g}\cap \{i\}|\\\label{eq:total_c_disjoint_intersection}
&\le& \sum_{i\in H} c \cdot 1 = c|H|.
\end{eqnarray}
Therefore,
$c\dis c G = |A| = \sum_{g\in A} 1 \leq \sum_{g\in A}|\supp{g}\cap\supp{g'}|\le cd(G')$, proving Eq.~\eqref{eq:dmin_bound}.

Similarly, Eq.~\eqref{eq:dmax_bound} follows from Eqs.~(\ref{eq:total_c_disjoint_intersection_0}-\ref{eq:total_c_disjoint_intersection}) by setting $H=[N]$ and using $|\supp{g}|\ge \min_{p\in G} |\supp p| = d(G)$.

To get (i) from Eqs.~(\ref{eq:triv_bound}-\ref{eq:dmax_bound}), note that they each hold for all $c$, and so we can replace $\dis c G$ with $\max_{c\ge1}\dis c G$ in all three equations. Since Eq.~\eqref{eq:triv_bound} also holds for all $G\in\mathcal{L}$, minimizing it over $G$ immediately implies $1\le\disj$ as well. In Eq.~\eqref{eq:dmin_bound} take $G'\in\mathcal{L}$ such that $d(G')=\dmin$ (and $G$ to be any anti-commuting logical Pauli) and in Eq.~\eqref{eq:dmax_bound} take $G\in\mathcal{L}$ so that $d(G)=\dmax$ to conclude $\disj\le \max_{c\ge1}\dis c G\le\dmin$ and $\disj\le\max_{c\ge1}\dis c G\le N/\dmax$, respectively.

We now prove (ii).
First, note that the implication $\dmin=1 \Longrightarrow\disj=1$ follows from (i). To show $\dmin>1 \Longrightarrow \disj>1$, we establish a stronger fact: for all $G\in\mathcal{L}$, if $\dmin>1$, then there exists $1\leq c \leq d(G)$ such that $\dis c G>1$. We make use of the following version of the Cleaning Lemma.
\begin{lem}[Cleaning Lemma \cite{Bravyi2009,Yoshida2010}]
\label{lem:cleaning}
For any non-trivial logical operator $G\in\mathcal{L}$ and any collection $R\subseteq[N]$ of qudit subsets $Q_i$ such that $|R|<\dmin$, there exists a representative $g\in G$ not supported on $R$, i.e., $\supp{g}\cap R=\emptyset$.
\end{lem}
\noindent Suppose $g$ is a minimal weight representative of $G$ and set $H=\supp{g}$.
Without loss of generality, we assume $H=[d(G)]$ (which might involve relabeling the qudit subsets $Q_i$).
For any $i\in H$, Lemma~\ref{lem:cleaning} and the assumption $\dmin>1$ guarantee we can find $g_i\in G$ that is not supported on the qudit subset $Q_i$, i.e, $i\not\in\supp{g_i}$.
We choose all distinct representatives $g,g_1,\ldots,g_{d(G)}$ of $G$ to form a set $A$.
By construction, there are at most $|A|-1$ elements of the set $A$ intersecting at any qudit subset $Q_i$.
Namely, if $i\in H$, then $i\not\in\supp{g_i}$, whereas if $i\not\in H$, then $i\not\in\supp g$.

Thus, the set $A$ can serve as an example of a $c$-disjoint subset of $G$ for $c=|A|-1$, and we obtain a lower bound $\dis c G \geq |A|/c > 1$ on the $c$-disjointness of $G$.
This, in turn, implies that the disjointness $\disj$ of the code is greater than one, $\disj >1$, finishing the proof of (ii).
\end{proof}

\noindent Certain codes even have disjointness saturating the upper bound in Lemma~\ref{lem:properties_of_disjointness}(i), as in the following example.
\begin{example}\label{ex:q_reed_muller}
Consider the family of Reed-Muller codes $\llbracket n =2^{D+1}-1,k=1\rrbracket$ for $D\geq 2$, which coincides with a family of color codes of distance three in $D$ spatial dimensions \cite{Steane1999,Anderson2014, Kubica2015a}.
We consider the single-qubit partition.
The two smallest codes in this family correspond to the $7$-qubit Steane and the $15$-qubit Reed-Muller codes.
The distance of logical $\bar X, \bar Y, \bar Z$ operators satisfies $d(\bar X)=d(\bar Y)=2^{D}-1$ and $d(\bar Z)=3$.
Thus, $\dmin = 3$ and $\dmax = 2^{D}-1$.
There are $2^{D+1}$ representatives of $\bar X$ and $2^{D+1}-1$ of them have minimal support.
The set of minimal representatives of $\bar X$ is, in fact, $\dmax$-disjoint, and therefore $\Delta_{\dmax}(\bar X)=n/\dmax$.
Moreover, for each representative $g$ of $\bar X$ one can always find at least one representative of $\bar Z$ (and thus of $\bar Y$) supported on $\supp g$.
We obtain that $\Delta_{\dmax}(\bar Z),\Delta_{\dmax}(\bar Y)\ge\Delta_{\dmax}(\bar X)=n/\dmax$, which results in a bound on the disjointness $\Delta\ge n/\dmax$.
However, $\disj \leq n/\dmax$ from Lemma~\ref{lem:properties_of_disjointness}(i), implying $\Delta = n/\dmax$.
\end{example}

The $c$-disjointness $\dis c G$ of a non-trivial logical operator $G\in\mathcal{L}$ quantifies how well $G$ can be ``cleaned'' (in the sense of \cite{Bravyi2009}) from an arbitrary subset of qudits. We conclude this section with a useful lemma needed to prove main results of our work.
\begin{lem}[scrubbing lemma]
\label{lem:scrubbing}
Consider a non-trivial logical operator $G\in\mathcal{L}$ and a collection $H\subseteq [N]$ of qudit subsets $Q_i$. For any $1\leq c\leq m^{n-k}$, there exists a representative $g\in G$ such that
\begin{equation}
\dis c G |\supp{g}\cap H|\le|H|.
\end{equation}
\end{lem}
\begin{proof}
Let $A\subseteq G$ be a maximal $c$-disjoint set, $|A| = c\dis c G$. Then,
\begin{eqnarray}
&&\dis c G |\supp{g}\cap H| = \frac{1}{c}|A|\min_{g\in A}|\supp{g}\cap H|\\
&&\leq \frac{1}{c} \sum_{g\in A} |\supp{g}\cap H| \le |H|,
\end{eqnarray}
where we use Eqs.~(\ref{eq:total_c_disjoint_intersection_0}-\ref{eq:total_c_disjoint_intersection}) for the second inequality.
\end{proof}
\noindent We note that if $\dis c G =1$, then the bound in Lemma~\ref{lem:scrubbing} is trivial, $|\supp{g}\cap H|\le|H|$. We get a non-trivial bound whenever $\dis c G >1$, which is exactly the situation for error-detecting stabilizer codes, see  Lemma~\ref{lem:properties_of_disjointness}(ii).

\section{Limitations on transversal gates}\label{sec:limit_trans}

In this section, we use the disjointness to bound the transversal logical gates on any error-detecting stabilizer code to the Clifford hierarchy of the logical Pauli group $\overline C_M = C_M(\overline{\mathcal{P}})$.
We start with a theorem for transversal operators on a single codeblock, which we later generalize to operators between $r$ codeblocks.

\setcounter{thm_counter}{\value{thm}}
\begin{thm}
\label{thm:finite_code_bound}
Consider a stabilizer code with min-distance $d_{\downarrow}$, max-distance $d_{\uparrow}$, and disjointness $\disj$.
If $M$ is an integer satisfying 
\begin{equation}
\dmax < \dmin \disj^{M-1},
\label{eq_maxlevel}
\end{equation}
then all transversal logical operators are in the $M^{\text{th}}$ level of the Clifford hierarchy $\overline C_M$.
\end{thm}

\begin{proof}
Let $G_j\in\mathcal{L}$ be any non-trivial logical Pauli operator, and let $K_0$ be a transversal logical operator.
We choose a representative $g_1$ of $G_1$ to have minimal support, $|\supp{g_1} | = d(G_1)$.
For $j\geq 1$, we recursively define $K_j=[K_{j-1},g_j]$, which is a transversal logical operator, and find $g_{j+1}\in G_{j+1}$ satisfying Lemma~\ref{lem:scrubbing} with $H=\supp{K_{j}}$.
Notice that bounding the support of the group commutator of two transversal operators $U_1,U_2$ is especially simple
\begin{equation}\label{eq:transversal_commutator}
\supp{[U_1,U_2]}\subseteq\supp{U_1}\cap\supp{U_2},
\end{equation}
which leads to the following bound for $j>1$
\begin{align}\label{eq:support_inequality}
|\supp{K_j}|&\le|\supp{K_{j-1}}\cap\supp{g_j}|\\
&\le |\supp{K_{j-1}}|/\dis{c_j}{G_j},
\end{align}
where the first and second inequalities were obtained by using Eq.~\eqref{eq:transversal_commutator} and Lemma~\ref{lem:scrubbing}, respectively.
Since we may choose arbitrary $c_j$, we set $c_j=\text{argmax}_{c\geq 1}\dis{c}{G_j}$. Now, using \eqref{eq:support_inequality} recursively, we find
\begin{eqnarray}
|\supp{K_M}| &\le& |\supp{K_1}| \prod_{j=2}^M\hspace{-2pt}\dis{c_j}{G_j}^{-1} \\ 
&\leq& \dmax /\disj^{M-1} \leq \dmin,
\end{eqnarray}
where in the the second inequality we used $|\supp{K_1}|\le |\supp{g_1}| \leq d(G_1) \leq \dmax$ and $\disj \geq \dis{c_j}{G_j}$. 
Since $|\supp{K_M}|$ is smaller than the min-distance $\dmin$ of the code, $K_M$ has to be a trivial logical operator.
Therefore, by definition of the Clifford-hierarchy, we recursively obtain that $K_{M-j}$ is a logical operator from the $j^{\text{th}}$ level. 
In particular, $K_0$ must be in the $M^{\text{th}}$ level $\overline C_M$.
\end{proof}

We remark that Theorem~\ref{thm:finite_code_bound} implies that transversal operators on a single codeblock of any error-detecting code must be in a finite level of the Clifford hierarchy.
Namely, from Lemma~\ref{lem:properties_of_disjointness}(ii) we get $\disj >1$, and thus we can always find an integer $M = \lceil \log_\disj (\dmax/\dmin) \rceil$ satisfying Eq.~\eqref{eq_maxlevel}.
We illustrate Theorem~\ref{thm:finite_code_bound} with the following examples.

\begin{example}
The non-CSS 5-qubit stabilizer code \cite{Bennett1996,Laflamme1996} has the stabilizer group $\mathcal{S}=\langle Z_1Z_2X_3X_5,X_1Z_2Z_3X_4,X_2Z_3Z_4X_5,X_1X_3Z_4Z_5\rangle$ and logical Pauli representatives $\bar X=X^{\otimes5}$ and $\bar Z=Z^{\otimes5}$ has $d_{\uparrow}=d_{\downarrow}=3$ and $\disj =5/3$ with respect to the single-qubit partition. Thus, $d_{\uparrow}<d_{\downarrow}\disj$ and so transversal logical gates must be in the Clifford group. In fact, the 5-qubit code has a transversal logical Clifford gate $SH$.
\end{example}

\begin{example}\label{ex:q_reed_muller_2}
As we already discussed in Example~\ref{ex:q_reed_muller}, the Reed-Muller code $\llbracket n = 2^{D+1}-1,k=1\rrbracket$ has parameters $\dmin=3$, $\dmax=2^{D}-1$ and $\Delta=n/\dmax$.
Thus, Theorem~\ref{thm:finite_code_bound} implies that the code can have transversal logical gates from at most the $M^{\text{th}}$ level of the Clifford hierarchy, where $M = \lceil \log_\disj (\dmax/\dmin) \rceil = D$.
In fact, the Reed-Muller code saturates this bound for any $D\geq 2$, since it has a transversal logical $\bar R_{D} = \mathrm{diag}(1,e^{2\pi i/ 2^{D}})$ gate.
\end{example}
\begin{example}
Depending on the qubit partition of the $\llbracket 105,1\rrbracket$ code from Example~\ref{ex:105_code}, its parameters are:
(a) $\dmin=3$, $\dmax=7$ and $\Delta=15/7$ or (b) $\dmin=\dmax=3$ and $\Delta=7/3$.
Thus, Theorem~\ref{thm:finite_code_bound} limits transversal logical gates with respect to the qubit partition to: (a) the third level of the Clifford hierarchy and (b) the Clifford group.
We emphasize that the transversal gates on the $\llbracket 105,1\rrbracket$ code saturate those bounds \cite{JL14}.
\end{example}

It is possible to treat multiple codeblocks (these need not even be the same code) as one large effective code. If the $b^{\text{th}}$ codeblock has partition $\{Q_i^{(b)}\}$, one can define a partition $\{Q_i\}$ of the effective code with each $Q_i$ consisting of (at most) one subset $Q_i^{(b)}$ from each codeblock. Moreover, if the partitions of each codeblock have distance $\dmin^{(b)}>1$, so too will the partition of the effective code have $\dmin>1$. 
Then, applying Theorem~\ref{thm:finite_code_bound} to the effective code leads to the following corollary.
\begin{cor}
Transversal gates on error-detecting stabilizer codes must be in the Clifford hierarchy.
\end{cor}
\noindent This in turn implies that the group of transversal logical gates on stabilizer codes is finite and not universal, providing an alternative proof of the main result of \cite{ZCC11}.

There are subtleties with this simple argument for multi-codeblock operators. First, it leaves the possibility that the achievable level of the Clifford hierarchy might depend on the number of considered codeblocks.
Second, the bound on level is not conveniently stated in terms of $d_{\downarrow},d_{\uparrow},\disj$ of the base code, but rather of the effective multiblock code.
We address both of these issues in Appendix~\ref{app:multiple_codeblock} with more detailed arguments for the multi-codeblock case.
We summarize the results with the following version of Theorem~\ref{thm:finite_code_bound} for stabilizer codes with multiple codeblocks. 

\begin{thm}[multi-codeblock case]
\label{thm:multiblock_code_bound_main}
Consider an $\llbracket n,k\rrbracket$ stabilizer code constructed from $m$-dimensional qudits. With respect to a partition of the qudits into $N$ subsets $Q_i$, let the code's parameters be $\dmin$, $\dmax$, and $\disj$. Now, consider $r$ codeblocks of this code, and let $r'=\min(r,N!m^{n-k})$. If $M$ is an integer satisfying
\begin{equation}\label{eq:multiblock_code_bound_main}
r'd_{\uparrow}\left(1-(1-1/\disj)^{r'}\right)^{M-1}<d_{\downarrow},
\end{equation}
then all transversal logical operators on $r$ codeblocks are in the $M^{\text{th}}$ level of the hierarchy $\overline C_M$.
\end{thm}

We remark that Theorem~\ref{thm:multiblock_code_bound_main} can do more than rule out universal sets of transversal operators. Any set of operators that is capable of bootstrapping itself up the Clifford hierarchy indefinitely also cannot be transversally implemented. A simple example is the Toffoli gate, which can be used to implement an $M$-qubit controlled-$X$ gate $C^{M-1}X$ for any $M$. Since the $C^{M-1}X$ gate is in the $M^{\text{th}}$ level, no stabilizer code can implement the Toffoli gate transversally; see Appendix~\ref{app:multiple_codeblock} for more details. We remark that the same limitation on the transversal Toffoli gate was recently proved for most quantum codes by using entirely different means \cite{Newman2017}.

Finally, a further generalization of Theorem~\ref{thm:finite_code_bound} comes by considering logical operators $K_0 = UP$ that can be written as a product of a transversal operator $U$ and a permutation $P$ of the subsets $Q_i$ (allowing for a different permutation on every  codeblock).
We can similarly restrict such logical operators to the Clifford hierarchy; see Appendix~\ref{app:permuting_transversal}.
However, for $r>1$ these logical operators do not form a group, so there is no obvious analog of the Eastin-Knill theorem \cite{EK09} for them.

\section{Limitations on shallow circuits}\label{sec:bounding_constant_depth_gates}

Our methods are powerful enough to put limitations on transversal as well as shallow-depth circuits which implement logical operators on stabilizer codes with respect to the given qudit partition.
In this section, we find bounds on the level of the Clifford hierarchy achievable by $q$-local circuits of depth $h$ (which may be geometrically non-local).
The key ingredient needed to derive explicit bounds in terms of parameters of the code ($\dmin$, $\dmax$, $\disj$) and of the circuit ($q$, $h$) is the following lemma.

\begin{lem}\label{lem:deep_circ_support_bound}
Let $A$ be a transversal operator and $U$ be a $q$-local circuit of depth $h$. Then,
\begin{equation}\label{eq:deep_circ_support_bound}
|\supp{[U,A]}|\le q^h|\supp{U}\cap\supp{A}|.
\end{equation}
\end{lem}

\begin{proof}
First, we express the transversal operator $A = \prod_{i\in \supp{A}} A_i$ as a product of operators $A_i$, each of which is supported only on one of the qudit subset, i.e., $|\supp{A_i}| = 1$.
Then, $\supp{A^\dag_i} \subseteq \supp{U A_i U^\dag}$ and $|\supp{U A_i U^\dag}| \leq q^h$.
Let $\mathcal{I} = \supp U \cap \supp A$, and then
$[U,A] = \left(\prod_{i\in \mathcal{I}} UA_i U^\dag\right) \prod_{i\in \mathcal{I}} A^\dag_i$.
Note that for any two operators $V$ and $W$ we have $\supp{VW} \subseteq \supp V \cup \supp W$.
Using this fact we get 
\begin{eqnarray}
\supp{[U,A]} &\subseteq& \bigcup_{i\in \mathcal{I}} \supp{UA_i U^\dag} \cup \bigcup_{i\in \mathcal{I}} \supp{A^\dag_i}\\
&=& \bigcup_{i\in \mathcal{I}} \supp{UA_i U^\dag},
\end{eqnarray}
and then using the union bound we arrive at
\begin{eqnarray}
|\supp{[U,A]}| &\leq& \sum_{i\in \mathcal{I}} |\supp{UA_i U^\dag}| \leq |\mathcal{I}| q^h.
\end{eqnarray}
This finishes the proof, since $ |\mathcal{I}| = |\supp U \cap \supp A|$.
\end{proof}

With Lemma~\ref{lem:deep_circ_support_bound}, we update Eq.~\eqref{eq:support_inequality} to read
\begin{align}\label{eq:deep_circ_support_inequality}
|\supp{K_j}|&\le q^{h_{j-1}}|\supp{K_{j-1}}\cap\supp{g_j}|\\
&\le q^{h_{j-1}}|\supp{K_{j-1}}|/\dis{c_j}{G_j},
\end{align}
where $h_{j-1}=2^{j-1}h$ is an upper bound on the depth of $K_{j-1}$.\footnote{Since $K_j=K_{j-1}g_jK_{j-1}^\dag g_j^\dag$ one immediately obtains $h_j\le 2h_{j-1}+2$, where $h_0=h$. However, the transversal operator $g_j$ can be absorbed into neighboring gates in the circuit and, as a result, does not increase to circuit depth. Thus, we can remove the additive constant from the recursion.} 
Accordingly, by repeating the argument recursively, we obtain a version of Theorem~\ref{thm:finite_code_bound} for $q$-local circuits of depth $h$.

\begin{thm}[shallow circuit case]
\label{thm:deep_circ_code_bound}
Consider a stabilizer code with min-distance $d_{\downarrow}$, max-distance $d_{\uparrow}$, and disjointness $\disj$. If $M$ is an integer satisfying
\begin{equation}\label{eq:deep_circ_code_bound}
d_{\uparrow} q^{(2^M-1)h}<d_{\downarrow}\disj^{M-1},
\end{equation}
then all logical operators implemented by $q$-local circuits of depth $h$ are in the $M^{\text{th}}$ level of the hierarchy $\overline C_M$.
\end{thm}

We remark that, unlike in Theorem~\ref{thm:finite_code_bound} for transversal operators, there is no guarantee that there exists $M$ satisfying Eq.~\eqref{eq:deep_circ_code_bound} for $q>1$.
Nevertheless, the shallow circuit version Theorem~\ref{thm:deep_circ_code_bound} is still useful for bounding logical gates on code \emph{families} in the asymptotic limit.
Namely, consider a family of codes $\llbracket n(l),k(l)\rrbracket$ with parameters $\dmin(l)$, $\dmax(l)$ and $\disj(l)$ with respect to some qudit partitions, parametrized by a positive integer $l$. We say that the code family has a $q$-local logical gate of depth $h$ if there exists a constant $l_0$ such that for all $l\ge l_0$ one can implement the logical gate in the corresponding codes with some $q$-local circuits of depth $h$. To rule out logical gates from outside the $M^{\text{th}}$ level of the hierarchy $\overline C_M$ with constant depth $h=h(l)$ and constant locality $q=q(l)$, it is therefore sufficient to consider the limit of Eq.~\eqref{eq:deep_circ_code_bound}. We arrive at the following corollary.
\begin{cor}\label{cor:asymptotic_code_bound}
If for a family of stabilizer codes $\llbracket n(l),k(l)\rrbracket$ with parameters $\dmin(l)$, $\dmax(l)$ and $\disj(l)$
there exists an integer $M$ satisfying 
\begin{equation}
\lim_{l\rightarrow\infty}\frac{\dmax(l)}{\dmin(l)\disj(l)^{M-1}}=0,
\end{equation}
then for any constants $q$ and $h$ all $q$-local logical gates of depth $h$ are in the $M^{\text{th}}$ level of the hierarchy $\overline C_M$.
\end{cor}
\noindent We require the limit vanish with $l$ (rather than, say, just being less than $1$) so that we can ignore the factors of constant locality and depth that appear in Eq.~\eqref{eq:deep_circ_code_bound}.

We conclude this section with a few examples illustrating the usefulness of Corollary~\ref{cor:asymptotic_code_bound}.
\begin{example}\label{ex:surface_code_2}
Consider the family of surface codes on square lattices of size $l\times l$.
As shown in Example~\ref{ex:surface_code}, the code parameters are $\dmin(l)=l$, $\dmax(l)=2l-1$, and $\disj(l)\ge l/2$.
Since for $M>1$ we have 
\begin{equation}
0 \leq \frac{\dmax(l)}{\dmin(l)\disj(l)^{M-1}}\leq 2^{M-1}\frac{2l-1}{l^M} \xrightarrow[l\rightarrow\infty]{} 0,
\end{equation}
thus constant-depth, constant-locality circuits on surface codes can only implement logical Clifford gates.
\end{example}
\noindent Surprisingly, asymmetric 2D codes can have transversal logical non-Clifford gates.
For instance, asymmetric Bacon-Shor codes have the transversal logical $CCZ$ gate \cite{Yoder2017}.
We emphasize that the asymmetry in the weight of different logical Pauli operators affects the ability to bound logical gates.

\begin{example}\label{ex:bacon_shor}
Consider the stabilizer code family of asymmetric Bacon-Shor codes in the $Z$-gauge on square lattices $l\times l^a$, $a\ge1$.
The code parameters $\dmin(l)=l$ and $\dmax(l)=l^a+l-1$ are asymptotically different. Similarly to Example~\ref{ex:surface_code}, we find $\disj(l)\ge l/2$. For $M>a$ we have
\begin{equation}
0 \leq \frac{\dmax(l)}{\dmin(l)\disj(l)^{M-1}}\le 2^{M-1}\frac{l^a+l-1}{l^M} \xrightarrow[l\rightarrow\infty]{} 0,
\end{equation}
and thus constant-depth, constant-locality logical circuits on asymmetric Bacon-Shor  codes are restricted to the ${(\lfloor a \rfloor +1)}^{\text{th}}$ level of the hierarchy $\overline C_{\lfloor a\rfloor +1}$.
\end{example}
\noindent The multi-block versions of the asymptotic arguments (taking the limit of Eq.~\eqref{eq:multiblock_code_bound_main}) in these two examples yield the same bounds.
 
One can also generalize Example~\ref{ex:surface_code_2} to other topological codes that are equivalent to the $D$-dimensional toric code, such as the color code \cite{Kubica2015b}. Choose logical Pauli $X$ and $Z$ operators to have representatives of dimensionality $D-s$ and $s$, where $1\leq s \leq \lfloor D/2 \rfloor$. Then, given linear lattice size $O(l)$, the code parameters are $\dmin = O(l^s)$, $\dmax = O(l^{D-s})$, $\disj = O(l^s)$, and thus from Corollary~\ref{cor:asymptotic_code_bound} their logical gates implemented via constant-depth (possibly geometrically non-local) circuits are limited to the $M^{\text{th}}$ level of the Clifford hierarchy, where $M=\lfloor(D-s)/s\rfloor+1$.
Note that as in Example~\ref{ex:bacon_shor}, the greater the asymmetry of the support of the logical operators (or, in other words, the difference in the dimensionality of those operators), the higher the level of the Clifford hierarchy that is accessible.
It is unclear though how to bound disjointness on more exotic topological codes with fractal-like logical operators, such as Haah's cubic code \cite{Haah2011}.

\section{Discussion}

We have provided explicit upper bounds on the level of the Clifford hierarchy that is accessible for logical operators on any stabilizer code, which are implemented by transversal and constant-depth circuits.
We expect our techniques to apply similarly to stabilizer codes composed of qudits, which differ in local dimension. 
As long as stabilizers and Pauli logical operators are tensor products of Pauli operators on physical qudits, presented results and proofs should carry through.

We remark that in the proof of Theorem~\ref{thm:finite_code_bound} instead of Lemma~\ref{lem:scrubbing} we could use the following simple corollary of the Cleaning Lemma~\ref{lem:cleaning}: for any non-trivial logical operator $G\in\mathcal{L}$ and a collection $H\subseteq [N]$ of qudit subsets $Q_i$ satisfying $|H|\le\dmin-1$, one can find a representative $g\in G$ such that $|H\cap\supp{g}|\le|H|-(\dmin-1)$. 
We follow the same recursive reduction of support of $K_j$ as in Theorem~\ref{thm:finite_code_bound} and obtain that if $M$ is an integer satisfying
\begin{equation}\label{eq:cleaning_bound}
\dmax<\dmin+(M-1)(\dmin-1),
\end{equation}
then all transversal logical gates are in the $M^{\text{th}}$ level of the Clifford hierarchy.
Such an integer always exists if the stabilizer code is error-detecting, i.e., $\dmin > 1$.
We note, that the bound on $M$ from Eq.~\eqref{eq:cleaning_bound} is rather loose.
In particular, transversal logical gates on asymmetric Bacon-Shor codes of size $O(l)\times O(l^2)$ are only restricted to the $O(l)^{\text{th}}$ level, which is not useful for large $l$.
On the other hand, Theorem~\ref{thm:finite_code_bound} limits the gates to the third level, which is indeed accessible in this code family, as we have seen in Example~\ref{ex:bacon_shor}. However, a strengthening of the bound Eq.~\eqref{eq:cleaning_bound} can be achieved by using any $M+1$ cleanable regions, each of which could potentially be supported on more than $\dmin-1$ qudits \cite{BeverlandPreskill}.

While our main results are derived without assumptions of geometric locality, we can derive even stronger bounds by assuming geometric locality of the circuits. For instance, $D$-dimensional surface codes (encoding a single logical qubit) cannot implement non-Clifford logical operators with geometrically local, constant-depth circuits. The argument follows exactly the same lines as that in Section~\ref{sec:intuition}, relying essentially on the ability to choose representatives $g_1,g_2$ of any two logical Paulis such that $|\supp{g_1}\cap\supp{g_2}|=O(1)$. Since a geometrically-local circuit $U$ cannot greatly distort the support of these representatives $|\supp{[[U,g_1],g_2]}|=O(1)$ as well. Note, however, that this argument breaks for geometrically local circuits that operate instead on several superimposed $D$-dimensional surface codes~\cite{Kubica2015b}, while Bravyi and K{\"o}nig's theorem would still hold in this case.

The notion of disjointness for stabilizer codes, which we introduced, appears to be difficult to calculate exactly.
If stabilizer codes have some underlying structure, as Reed-Muller codes in Example~\ref{ex:q_reed_muller} or topological codes in Examples~\ref{ex:surface_code_2}~and~\ref{ex:bacon_shor}, then we can find bounds on the disjointness, and this usually suffices to establish limits on the accessible level of the Clifford hierarchy.
We believe that it is a challenging open problem to find efficient methods to compute (or approximate) the disjointness for an arbitrary stabilizer code.
This problem, however, might be substantially simpler for topological codes, where one could exploit code and lattice symmetries.
Also, it would be interesting to extend the notion of disjointness to the subsystem codes and find possible relations to other new stabilizer code quantities, such as the price \cite{Pastawski2017}.

\section*{Acknowledgements}

The authors would like to thank Ben Brown, Steve Flammia and Daniel Gottesman for helpful discussions. In particular, we would like to thank Michael Beverland for comments on the manuscript and who also showed the transversal gates of all stabilizer and subsystem codes are restricted to the Clifford hierarchy in unpublished work with John Preskill~\cite{BeverlandPreskill}.
TJ acknowledges the support from the Walter Burke Institute
for Theoretical Physics in the form of the Sherman
Fairchild Fellowship.
AK acknowledges funding provided by the Simons Foundation through the ``It from Qubit'' Collaboration, as well as by the Institute for Quantum Information and Matter, an NSF Physics Frontiers Center (NFS Grant PHY-1125565) with support of the Gordon and Betty Moore Foundation (GBMF-12500028).
Research at Perimeter Institute is supported by the Government of Canada
through Industry Canada and by the Province of Ontario through the Ministry of Research and Innovation.
TY is grateful for support from the Department of Defense (DoD) through the National Defense Science and Engineering Graduate (NDSEG) Fellowship program and also an IBM PhD Fellowship award.

\appendix

\section{Logical gates on multiple codeblocks}\label{app:multiple_codeblock}
In this section, we describe how to restrict gates on multiple codeblocks to the Clifford hierarchy. We flesh out the argument in the main text by giving formulas for $d_{\downarrow},d_{\uparrow},\disj$ of the multi-codeblock code in terms of those for the the single codeblock. Then we argue how the bound on level of the Clifford hierarchy obtainable by transversal gates can be made independent of the number of codeblocks.

Consider $r$ codeblocks of the same\footnote{Taking the blocks to be the same is for simplification of the argument only. For instance, we do not have to deal with different quantities $d_{\downarrow},d_{\uparrow},\disj$ for each code. Running through a more general argument where the stabilizer codes are allowed to be different is possible, and similarly restricts transversal gates to the Clifford hierarchy.} $\llbracket n,k\rrbracket$ base stabilizer code, each with identical\footnote{Again, the sameness of the partitions can be relaxed at the cost of notational encumbrance.} qudit partitions, which we write as $\{Q^{(b)}_i\}$, where superscript $b=1,2,\dots,r$ represents the codeblock. Like in the main text, we say there are $N$ subsets $Q^{(b)}_i$ for each $b$, and we use $d_{\downarrow},d_{\uparrow},\disj$ to denote the quantities of the base code.

The effective stabilizer code is formed by treating all $r$ codeblocks as a $\llbracket rn,rk\rrbracket$ stabilizer code, and the qudits of the effective code can be partitioned into subsets $\{Q_i\}$, each consisting of one subset from each of the $r$ codeblocks,
\begin{equation}
Q_i=\bigcup_{b=1}^{r}Q_{\sigma_b(i)}^{(b)}.
\end{equation}
Here $\sigma_b:[N]\rightarrow[N]$ is an (arbitrary) permutation of the partitions of codeblock $b$. This completes the partitioning of the effective code in such a way that the effective code's min-distance  equals that of the base code, $d_{\downarrow,\text{eff}}=d_{\downarrow}$. We also note the simple bound on the effective code's max-distance $d_{\uparrow,\text{eff}}\le rd_{\uparrow}$.

The final quantity to address is the disjointness of the effective code $\disj_{\text{eff}}$. To do this, we prove a more general version of Lemma~\ref{lem:scrubbing} for multiple codeblocks, and let this inform the definition of $\disj_{\text{eff}}$. Start by establishing some notation. Let $\mathcal{L}_{\text{eff}}$ denote the set of nontrivial logical cosets of the effective code. Note that $G\in\mathcal{L}_{\text{eff}}$ means, by definition, that
\begin{equation}\label{eq:G_decomp}
G=\bigotimes_{b=1}^rG^{(b)},
\end{equation}
where $G^{(b)}\in\mathcal{L}\cup\{\mathcal{S}\}$ are logical cosets of the base code and at least one is nontrivial (i.e.~in $\mathcal{L}$).
\begin{lem}\label{lem:multiblock_scrubbing}
Let $G\in\mathcal{L}_{\text{eff}}$ and $H\subseteq[N]$. Then, for any $c_1,c_2,\dots,c_r$, there exists a representative $g\in G$ such that
\begin{equation}\label{eq:multiblock_scrubbing}
|\supp{g}\cap H|\le\left(1-\prod_{b}\left(1-\dis{c_b}{G^{(b)}}^{-1}\right)\right)|H|,
\end{equation}
where the product ranges only over nontrivial cosets in the decomposition of $G$, Eq.~\eqref{eq:G_decomp}.
\end{lem}
\begin{proof}
Without loss of generality we say that only the first $r_0\le r$ cosets in Eq.~\eqref{eq:G_decomp} are nontrivial. We decompose $g\in G$ as $g=\bigotimes_{b=1}^rg^{(b)}$ with $g^{(b)}\in G^{(b)}$. Our task is to find $g^{(b)}$ such that Eq.~\eqref{eq:multiblock_scrubbing} holds. Start by noting
\begin{align}
\supp{g}&=\bigcup_{b=1}^{r_0}\supp{g^{(b)}}\\
H\cap\supp{g}&=\bigcup_{b=1}^{r_0}\left(H\cap\supp{g^{(b)}}\right).
\end{align}
Say that we have already chosen $g^{(1)},g^{(2)},\dots,g^{(j-1)}$. Then we need only minimize the intersection of $g^{(j)}$ with
\begin{equation}
H_{j-1}:=H-\bigcup_{b=1}^{j-1}\left(H\cap\supp{g^{(b)}}\right),
\end{equation}
the set of partitions in $H$ that are yet unaffected. By Lemma~\ref{lem:scrubbing} we can find $g^{(j)}\in G^{(j)}$ such that
\begin{equation}\label{eq:apply_scrubbing}
|\supp{g^{(j)}}\cap H_{j-1}|\le|H_{j-1}|/\dis{c_j}{G^{(j)}}.
\end{equation}
Note the relations
\begin{align}
H_0&=H,\\
H_j&=H_{j-1}-(\supp{g^{(j)}}\cap H_{j-1}),\\
H-H_{r_0}&=H\cap\supp{g}.
\end{align}
Thus, Eq.~\eqref{eq:apply_scrubbing} implies
\begin{align}\label{eq:Hj_size_relation}
|H_j|&=|H_{j-1}-(\supp{g^{(j)}}\cap H_{j-1})|\\\nonumber
&\ge\left(1-\dis{c_j}{G^{(j)}}^{-1}\right)|H_{j-1}|.
\end{align}
Repetitive use of Eq.~\eqref{eq:Hj_size_relation} gives us the bound
\begin{equation}
|H_{r_0}|\ge\prod_{b=1}^{r_0}\left(1-\dis{c_b}{G^{(b)}}^{-1}\right)|H|
\end{equation}
from which we conclude
\begin{align}
|H-H_{r_0}|&=|H\cap\supp{g}|\\\nonumber
&\le\left(1-\prod_{b=1}^{r_0}\left(1-\dis{c_b}{G^{(b)}}^{-1}\right)\right)|H|.
\end{align}
This completes the proof.
\end{proof}

We can simplify Eq.~\eqref{eq:multiblock_scrubbing} by choosing specific $c_b$ such that $\disj \le \dis{c_b}{G^{(b)}}$ and find a $g\in G$ such that
\begin{align}
|H\cap\supp{g}|&\le(1-(1-1/\disj)^{r_0})|H|\\\label{eq:multiblock_scrubbing_with_tilde_d}
&\le (1-(1-1/\disj)^{r})|H|.
\end{align}
The latter form of the right-hand side implies that defining
\begin{equation}
\disj_{\text{eff}}:=\frac{1}{1-(1-1/\disj)^r}
\end{equation}
will result in a theorem analogous to Theorem~\ref{thm:finite_code_bound}, but for multiple codeblocks.
\begin{thm}\label{thm:multiblock_code_bound}
If $d_{\uparrow,\text{eff}}<d_{\downarrow,\text{eff}}\disj_{\text{eff}}^{M-1}$, then all transversal gates on $r$ codeblocks are in $\overline C_M$.
\end{thm}
\begin{proof}
This follows the same lines as the proof of Theorem~\ref{thm:finite_code_bound}, but using Eq.~\eqref{eq:multiblock_scrubbing_with_tilde_d} in place of Lemma~\ref{lem:scrubbing}.
\end{proof}
\noindent Of course, we can write the condition of Theorem~\ref{thm:multiblock_code_bound} solely in terms of the single codeblock parameters $d_{\downarrow},d_{\uparrow},\disj$. That is, if
\begin{equation}\label{eq:multiblock_condition_in_terms_of_single_block}
rd_{\uparrow}<d_{\downarrow}\frac{1}{\left(1-(1-1/\disj)^r\right)^{M-1}},
\end{equation}
then all transversal gates are in $\mathcal{C}_M$. Since $\disj>1$ for error-detecting stabilizer codes (Lemma~\ref{lem:properties_of_disjointness}(ii)), $\disj_{\text{eff}}>1$ as well, and the right-hand side of Eq.~\eqref{eq:multiblock_condition_in_terms_of_single_block} must exceed the left for some sufficiently large $M\ge M_0$.

However, given only the arguments until now, it is still possible that $M_0$ depends on $r$ and even that increasing $r$ arbitrarily can increase $M_0$ arbitrarily as well. This would imply that high-level transversal gates between different codeblocks are easier to find than transversal single block gates. While this may be true to some extent there is a limit, which we describe now.

The key is to realize in what instances we can find $g^{(j)}$ so that $|H_j|=|H_{j-1}|$ in Eq.~\eqref{eq:Hj_size_relation}. This happens when we can choose $g^{(j)}$ so that
\begin{equation}
H\cap\supp{g^{(j)}}\subseteq\bigcup_{b=1}^{j-1}H\cap\supp{g^{(b)}}.
\end{equation}
For instance, in the simple case when $\sigma_b$ are each the identity permutation, then whenever $G^{(b_1)}=G^{(b_2)}=G'$, we might as well choose the same representative $g'\in G'$ for both $g^{(b_1)}$ and $g^{(b_2)}$ because then
\begin{equation}
H\cap\supp{g^{(b_1)}}=H\cap\supp{g^{(b_2)}}.
\end{equation}
Moreover, we are guaranteed to start repeating cosets in the decomposition Eq.~\eqref{eq:G_decomp} when $r>m^{n-k}$, so effectively we can replace $r$ in Eq.~\eqref{eq:multiblock_condition_in_terms_of_single_block} with $\min(r,m^{n-k})$, thus achieving an $r$-independent bound.

When $\sigma_b$ is arbitrary, we can make the same argument when $G^{(b_1)}=G^{(b_2)}$ \emph{and} $\sigma_{b_1}=\sigma_{b_2}$. Since there are finitely many permutations as well, we can replace $r$ in Eq.~\eqref{eq:multiblock_condition_in_terms_of_single_block} with $\min(r,N!m^{n-k})$, which admittedly is large but at least finite. The previous arguments complete the proof of Theorem~\ref{thm:multiblock_code_bound_main}.

The upshot of these finite bounds on $M_0$ is that we can state further no-go theorems on what particular gates can be implemented on stabilizer codes. For instance,
\begin{cor}
No error-detecting stabilizer code (on qubits) can implement Toffoli transversally.
\end{cor}
\begin{proof}
There is a well-known construction \cite{Nielsen2010} where, for any integer $w\ge2$, $2w-3$ Toffoli gates and $2w-1$ qubits ($w-2$ of which are ancillas) suffice to make $\text{C}^wX$, i.e.~$X$ with $w$ control qubits.
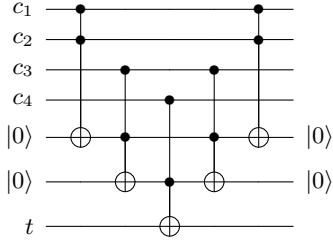
\begin{figure}
\begin{equation*}
\Qcircuit @C=1em @R=1em {
& \lstick{c_1}     & \ctrl{4}  & \qw      & \qw      & \qw      & \ctrl{4} & \qw                  \\ 
& \lstick{c_2}     & \ctrl{3}  & \qw      & \qw      & \qw      & \ctrl{3} & \qw                  \\ 
& \lstick{c_3}     & \qw       & \ctrl{3} & \qw      & \ctrl{3} & \qw      & \qw                  \\ 
& \lstick{c_4}     & \qw       & \qw      & \ctrl{3} & \qw      & \qw      & \qw                  \\ 
& \lstick{\ket{0}} & \targ     & \ctrl{1} & \qw      & \ctrl{1} & \targ    & \rstick{\ket{0}} \qw \\
& \lstick{\ket{0}} & \qw       & \targ    & \ctrl{1} & \targ    & \qw      & \rstick{\ket{0}} \qw \\
& \lstick{t}       & \qw       & \qw      & \targ    & \qw      & \qw      & \qw                  \\
} 
\end{equation*}
\caption{\label{fig:ClX} Making a $\text{C}^4X$ gate with controls $c_1,c_2,c_3,c_4$ and target $t$ from five Toffolis and two ancillas.}
\end{figure}
See Fig.~\ref{fig:ClX} for an example with $w=4$. Since $\text{C}^wX\in C_{w+1}$, we see that having transversal Toffoli would imply transversal gates in every level of the Clifford hierarchy. But this is ruled out by the finite bound on level argued for above.
\end{proof}
\noindent The same conclusion was shown for most quantum codes in \cite{Newman2017} by reduction to bounds in homomorphic encryption. Our proof technique can be applied to any other set of gates that is, like Toffoli, capable of bootstrapping itself indefinitely up the hierarchy.

\section{Transversal gates with permutations}\label{app:permuting_transversal}
In this section, we extend Theorem~\ref{thm:multiblock_code_bound_main} to the case of permuting transversal operators $K_0$, which are those that can be written as $K_0=UP$ for transversal $U$ and permutation $P$ of the partitions $\mathcal{H}_i$ separately for each codeblock.

As in the proof of Theorem~\ref{thm:finite_code_bound}, take an arbitrary sequence of cosets $G_1,G_2,\dots\in\mathcal{L}$ and define $K_1=[K_0,g_1]=UPg_1PU^\dag g_1^\dag$ and $K_j=[K_{j-1},g_j]$ for some choices of $g_j\in G_j$. The key thing to notice is that the recursive reduction of support of the $K_j$ is modified only at $K_1$. Take $g_1\in G_1\in\mathcal{L}$ to have minimal support $|\supp{g_1}|=d(G_1)$, so that
\begin{equation}\label{eq:doubled_support}
|\supp{K_1}|\le2d(G_1),
\end{equation}
simply because $Pg_1P$ may have disjoint support from $g_1^\dag$. Bounding the supports of $K_j$ can then be done exactly as in the proof of Theorem~\ref{thm:finite_code_bound}. More generally, the argument for Theorem~\ref{thm:multiblock_code_bound_main} found in Appendix~\ref{app:multiple_codeblock} can incorporate the observation Eq.~\eqref{eq:doubled_support} to show
\begin{thm}
Consider a stabilizer code with quantities $d_{\downarrow},d_{\uparrow},\disj$. Let $r'=\min(r,N!m^{n-k})$. If
\begin{equation}\label{eq:multiblock_code_bound}
2r'd_{\uparrow}\left(1-(1-1/\disj)^{r'}\right)^{M-1}<d_{\downarrow},
\end{equation}
then all permuting transversal gates on $r$ codeblocks are in $\overline C_M$. When $r=1$,
\begin{equation}
2d_{\uparrow}<d_{\downarrow}\disj^{M-1}
\end{equation}
implies the same for one codeblock.
\end{thm}
\noindent Notice that for single codeblocks $r=1$ the permuting transversal operators $K_0$ \emph{do} form a group, and thus this theorem has a corollary that the group of permuting transversal operators on a single codeblock is finite and non-universal.

\bibliography{references.bib}

\end{document}